\documentclass{elsarticle}
\usepackage{amssymb,amsfonts,amsmath,latexsym,dsfont}
\usepackage{tikz,pgflibraryplotmarks}
\usepackage{graphicx}
\usepackage{epstopdf}
\usepackage{mathrsfs}
\usepackage{fullpage}
\usepackage{hyperref}
\usepackage[boxed,noline]{algorithm2e}
\usepackage{color}
\usepackage{multirow}
\usepackage{subfigure}
\DeclareMathOperator*{\argmin}{arg\,min}

\graphicspath{{}{./}}

\usepackage{boxedminipage}

\newcommand{\hf}{{\frac 12}}

\newcommand{\cN}{{\cal N}}

\renewcommand{\div}{\nabla\cdot\,}
\newcommand{\grad}{\ensuremath{\nabla}}

\newcommand{\bfA}{{\bf A}}

\newcommand{\bfD}{{\bf D}}

\newcommand{\bfH}{{\bf H}}
\newcommand{\bfI}{{\bf I}}
\newcommand{\bfJ}{{\bf J}}

\newcommand{\bfP}{{\bf P}}

\newcommand{\bfe}{{\bf e}}
\newcommand{\bff}{{\bf f}}

\newcommand{\bfx}{{\bf x}}

\newcommand{\bfq}{{\bf q}}
\newcommand{\bfp}{{\bf p}}

\newcommand{\bfd}{{\bf d}}
\newcommand{\bfm}{{\bf m}}

\newcommand{\bftau}{{\boldsymbol \tau}}
\newcommand{\bfdelta}{{\boldsymbol \delta}}

\newcommand{\known}{{\tt\textit{known}}}
\newcommand{\front}{{\tt\textit{front}}}
\newcommand{\unknown}{{\tt\textit{unknown}}}

\newcommand{\obs}{{\sf {_{obs}}}}

\newcommand{\diag}{{\sf{diag}}\,}

\newtheorem{lemma}{Lemma}
\newtheorem{corollary}{Corollary}

\newtheorem{proof}{Proof}

\newcommand{\mulcol}[3]{\multicolumn{#1}{#2}{#3}}

\begin{document}

\begin{frontmatter}

\title{A fast marching algorithm for the factored eikonal equation}
%\author{Eran Treister\thanks{Department of Earth and Ocean Sciences, The University of British Columbia, Vancouver, BC, Canada. \newline emails: {\tt eran@cs.technion.ac.il, haber@math.ubc.ca}.}
%\and Eldad Haber\footnotemark[1] \thanks{Department of Mathematics, The University of British Columbia, Vancouver, BC, Canada.}}
%
\author[mymainaddress]{Eran Treister\corref{mycorrespondingauthor}}
\cortext[mycorrespondingauthor]{Corresponding author}
\ead{erant@bgu.ac.il}
\author[mymainaddress,mysecondaryaddress]{Eldad Haber}
\ead{haber@math.ubc.ca}
\address[mymainaddress]{Department of Earth and Ocean Sciences, The University of British Columbia, Vancouver, BC, Canada.}
\address[mysecondaryaddress]{Department of Mathematics, The University of British Columbia, Vancouver, BC, Canada.}
%
%\maketitle

\begin{abstract}
The eikonal equation is instrumental in many applications in several fields ranging from computer vision to geoscience. This equation can be efficiently solved using the iterative Fast Sweeping (FS) methods and the direct Fast Marching (FM) methods. However, when used for a point source, the original eikonal equation is known to yield inaccurate numerical solutions, because of a singularity at the source. In this case, the factored eikonal equation is often preferred, and is known to yield a more accurate numerical solution. One application that requires the solution of the eikonal equation for point sources is travel time tomography. This inverse problem may be formulated using the eikonal equation as a forward problem. While this problem has been solved using FS in the past, the more recent choice for applying it involves FM methods because of the efficiency in which sensitivities can be obtained using them. However, while several FS methods are available for solving the factored equation, the FM method is available only for the original eikonal equation.

In this paper we develop a Fast Marching algorithm for the factored eikonal equation, using both first and second order finite-difference schemes. Our algorithm follows the same lines as the original FM algorithm and requires the same computational effort. In addition, we show how to obtain sensitivities using this FM method and apply travel time tomography, formulated as an inverse factored eikonal equation. Numerical results in two and three dimensions show that our algorithm solves the factored eikonal equation efficiently, and demonstrate the achieved accuracy for computing the travel time. We also demonstrate a recovery of a 2D and 3D heterogeneous medium by travel time tomography using the eikonal equation for forward modeling and inversion by Gauss-Newton.
\end{abstract}

\end{frontmatter}

\section{Introduction}

The eikonal equation appears in many fields, ranging from computer vision \cite{sethian1996fast,sethian1999fast,spira2004efficient,kimmel1998computing}, where it is used to track evolution of interfaces, to geoscience \cite{luo2012higher,leung2007eulerian,qian2002adaptive,li2013first,rawlinson2004wave} where it describes the propagation of the first arrival of a wave in a medium. The equation has the form
\begin{equation}\label{eq:eikonal}
|\nabla\tau|^2 = \kappa(\vec{x})^2,
\end{equation}
where $|\cdot|$ is the Euclidean norm. In the case of wave propagation, $\tau$ is the travel time of the wave and $\kappa(\vec{x})$ is the slowness (inverse velocity) of the medium. The value of $\tau$ is usually given at some sub-region. For example, in this work we assume
the wave propagates from a point source at location $\vec{x}_0$, for which the travel time is 0, and hence $\tau(\vec{x}_{0}) = 0$.

Equation \eqref{eq:eikonal} is nonlinear, and as such may have multiple branches in its solution. One of these branches, which is the one of interest in the applications mentioned earlier, corresponds to the "first-arrival" viscosity solution, and can be calculated efficiently \cite{crandall1983viscosity,rouy1992viscosity}. One of the ways to compute it is by using the Fast Marching (FM) methods \cite{tsitsiklis1995efficient,sethian1996fast,sethian1999fast}, which solve it directly using first or second order schemes in ${\cal O}(n\log n)$ operations. These methods are based on the monotonicity of the solution along the characteristics. Alternatively, \eqref{eq:eikonal} can be solved iteratively by Fast Sweeping (FS) methods, which may be seen as Gauss-Seidel method for \eqref{eq:eikonal}. First order accurate solutions of \eqref{eq:eikonal} can be obtained very efficiently in $2^d$ Gauss-Seidel sweeps in ${\cal O}(n)$ operations, where $d$ is the dimension of the problem \cite{tsai2003fast, zhao2005fast}. An alternative for the mentioned approaches is to use FS to solve a Lax-Friedrichs approximation for \eqref{eq:eikonal}, which involves adding artificial viscosity to the original equation \cite{kao2004lax}. This approach was suggested for general Hamilton-Jacobi equations, and is simple to implement. In \cite{qian2007fast}, such Lax-Friedrichs approximation is obtained using FS up to third order accuracy using the weighted essentially non-oscillatory (WENO) approximations to the derivatives. For a performance comparison between some of the mentioned solvers see \cite{gremaud2006computational}.

In some cases, the eikonal equation \eqref{eq:eikonal} is used to get a geometrical-optics ansatz of the solution of the Helmholtz equation in the high frequency regime \cite{leung2007eulerian,luo2011factored,luo2012higher,luo2014fast}. This is done using the Rytov decomposition of the Helmholtz solution: $u(\vec{x}) = a(\vec{x})\exp(i\omega\tau(\vec{x}))$, where $a(\vec{x})$ is the amplitude and $\tau(\vec{x})$ is the travel time. This approach involves solving \eqref{eq:eikonal} for the travel-time and solving the transport equation
\begin{equation}\label{eq:transport}
\nabla\tau\cdot\nabla a + \frac{1}{2}a\Delta\tau = \frac{1}{2}(\nabla\tau\cdot\nabla a + \div (a\grad\tau)) = 0
\end{equation}
for the amplitude \cite{luo2011factored,luo2012higher}. The resulting approximation includes only the first arrival information of the wave propagation. Somewhat similarly, the work of \cite{haber2011fast} suggests using the eikonal solution to get a multigrid preconditioner for solving linear systems arising from discretization of the Helmholtz equation. %There, unlike the case of \eqref{eq:transport}, the resulting amplitude contains all the information of the wave propagation, e.g. it includes reflections.

In many cases in seismic imaging, the eikonal equation is used for modeling the migration of seismic waves from a point source at some point $\vec{x}_0$. In this case, when solving \eqref{eq:eikonal} numerically by standard finite differences methods, the obtained solution has a strong singularity at the location of the source, which leads to large numerical errors \cite{qian2002adaptive,fomel2009fast}. Figure \ref{fig:badSource} illustrates this phenomenon by showing the approximation error for the gradient of the distance function, which is the solution of \eqref{eq:eikonal} for $\kappa=1$. It is clear that the largest approximation error for the gradient is located around the source point, and that its magnitude is rather large. In addition, it is observed that although $\tau$ may have more singularities in other places, the singularity at the source is more damaging and polluting for the numerical solution \cite{qian2002adaptive,fomel2009fast}.

\begin{figure}
   \centering
   \includegraphics[width=0.4\textwidth]{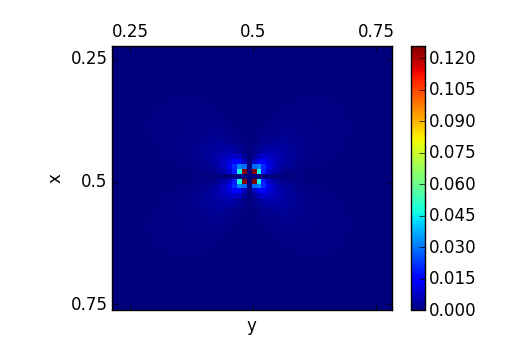}
   \caption{The $l_2$ norm of the approximation error $|\nabla\tau_0 - D\tau_0|$ around a source point at [0.5,0.5], where $\tau_0$ is the distance function and $D$ is a central difference gradient operator on a mesh with $h_x = h_y=0.01$.}
    \label{fig:badSource}
\end{figure}

A rather easy treatment to the described phenomenon is suggested in \cite{pica1997fast,fomel2009fast}, and achieved using a factored version of \eqref{eq:eikonal}. That is, we define
\begin{equation}\label{eq:factoredTau}
\tau = \tau_0\tau_1,
\end{equation}
where $\tau_0$ is the distance function, $\tau_{0} = \|\vec{x}-\vec{x}_0\|_2$, from the point source---the analytical solution for \eqref{eq:eikonal} in the case where $\kappa(\vec{x})=1$ is a constant. Indeed, at the location of the source, the function $\tau_0$ is non-smooth. However, the computed factor $\tau_1$ is expected to be very smooth at the surrounding of the source, and can be approximated up to high order of accuracy \cite{LouQianBurridge2014}. Plugging \eqref{eq:factoredTau} into \eqref{eq:eikonal} and applying the chain rule yields the \emph{factored eikonal equation}
\begin{equation}\label{eq:factoredeikonal1}
\tau_0^2|\nabla\tau_1|^2 + 2\tau_0\tau_1\nabla\tau_0\cdot\nabla\tau_1 + \tau_1^2 = \kappa(\vec{x})^2.
\end{equation}
Similarly to it original version, Equation \eqref{eq:factoredeikonal1} can be solved by fast sweeping methods in first order accuracy \cite{fomel2009fast,luo2012fast,LouQianBurridge2014}, or by a Lax-Friedrichs scheme up to third order of accuracy \cite{luo2012higher,luo2011factored,LouQianBurridge2014}. The recent works \cite{LouQianBurridge2014,noble2014accurate} suggest hybrid schemes where the factored eikonal is solved at the neighborhood of the source and the standard eikonal, which is computationally easier, is solved in the rest of the domain.

One geophysical tool that fits the scenario described earlier is travel time tomography. One way to formulate it is by using the eikonal equation as a forward problem inside an inverse problem \cite{sei1994gradient}. To solve the inverse problem, one should be able to solve \eqref{eq:eikonal} accurately for a point source, and to compute its sensitivities efficiently. The works of \cite{leung2006adjoint,taillandier2009first} computes the tomography by FS, and require an FS iterative solution for computing the sensitivities. The more recent \cite{li2013first} uses the FM algorithm for forward modeling using the non-factored eikonal equation, because this way the sensitivities are obtained more efficiently by a simple solution of a lower triangular linear system. \cite{benaichouche2015first} suggests to use FS for forward modeling using the factored equation, but also efficiently obtain the sensitivities by approximating them using FM with the non-factored eikonal equation.

In this paper, we develop a Fast Marching algorithm for the factored eikonal equation \eqref{eq:factoredeikonal1}, based on \cite{sethian1996fast,sethian1999fast}. As in \cite{sethian1999fast}, our algorithm is able to solve \eqref{eq:factoredeikonal1} using first order or second order schemes, in guaranteed ${\cal O}(n\log n)$ running time. When using our method for forward modeling in travel time tomography, one achieves both worlds: (1) have an accurate forward modeling based on the factored eikonal equation, and (2) obtain the sensitivities of the (factored) forward modeling efficiently, by solving lower triangular linear systems in ${\cal O}(n)$ operations.
Computationally, this is one of the most attractive ways to solve the inverse problem, since the cost of the inverse problem can be governed by the cost of applying the sensitivities.

\bigskip

Our paper is organized as follows. In the next section we briefly review the FM method in \cite{sethian1996fast,sethian1999fast}, including some of its implementation details. Next, we show our extension to the FM method for the factored eikonal equation---in both first and second order of accuracy---and provide some theoretical properties. Following that, we discuss the derivation of sensitivities using FM and briefly present the travel time tomography problem. Last, we show some numerical results that demonstrate the effectiveness of the
method in two and three dimensions, for both the forward and inverse problems.

\section{The Fast Marching algorithm}

We now review the FM algorithm of \cite{sethian1996fast,sethian1999fast} in two dimensions. The extension to higher dimensions is straightforward. The FM algorithm is based on the Godunov upwind discretization \cite{rouy1992viscosity} of \eqref{eq:eikonal}. In two dimensions, this discretization is given by
\begin{equation}\label{eq:Godunov1}
|\nabla\tau|^2 \approx \left[\max\{D^{-x}_{ij}\tau,-D^{+x}_{ij}\tau,0\}^2 + \max\{D^{-y}_{ij}\tau,-D^{+y}_{ij}\tau,0\}^2\right]=\kappa(\vec{x}_{ij})^2,\quad \vec{x}_{ij}\in\Omega_h,
\end{equation}
where in the simplest form $D^{-x}_{ij}\tau = \frac{\tau_{i,j}-\tau_{i-1,j}}{h}$ and $D^{+x}_{ij}\tau = \frac{\tau_{i+1,j}-\tau_{i,j}}{h}$ are the backward and forward first derivative operators, respectively.
In principal, one can replace these operators with ones of higher order of accuracy.

The FM algorithm solves \eqref{eq:Godunov1} in a sophisticated way, exploiting the fact that the upwind difference structure of \eqref{eq:Godunov1} imposes a unique direction in which the information propagates---from smaller values of $\tau$ to larger values. Hence, the FM algorithm rests on solving \eqref{eq:Godunov1} by building the solution outwards from the smallest $\tau$ value. It assumes that some initial value of $\tau$ is given at some region of $\Omega_h$ (or a point $\vec{x}_0$) and it propagates outwards from this initial region, by updating the next smallest value of $\tau$ at each step.

To apply the rule above, let us define three disjoint sets of variables: the \known variables, the \front variables (which are sometimes called the \emph{trial} variables) and the \unknown variables. These three sets together contain all the grid points in the problem. For simplicity, let us assume that we solve \eqref{eq:Godunov1} for a point source. That is, a source is located at point $\vec{x}_0$, for which $\tau(\vec{x}_0) = 0$. Initially, \known is chosen as an empty set, \front is set to contain only $\vec{x}_0$, and \unknown has the rest of the variables for which $\tau$ is set to infinity. At each step we choose the point $\vec{x}_{ij}$ in \front with minimal value of $\tau(\vec{x}_{ij})$ and move it to \known. Next, we move all of its neighbors which are in \unknown to \front, and solve \eqref{eq:Godunov1} for all neighbors which are not in \known. This way, we set all variables to be in \known in $n$ steps, and the algorithm finishes. A precise description of the algorithm is given in Algorithm \ref{alg:FM}.

\begin{algorithm}
\DontPrintSemicolon
\label{alg:FM}
\KwSty{Algorithm: Fast Marching }\;
Initialize: \;
$\tau_{ij} = \infty$ for all $\vec{x}_{ij}\in\Omega_h$, $\tau(\vec{x}_0) = 0$,\;
\known$\leftarrow\emptyset$, \front$\leftarrow\{\vec{x}_0\}$.\;
\While{\front$\neq\emptyset$}{
\begin{enumerate}
    \item \label{step:find_minimum}Find the minimal entry in \front:\;
    $\quad\quad \vec{x}_{i_{min},j_{min}} = \argmin_{\vec{x}_{ij}}\{\tau_{ij}:\vec{x}_{ij}\in$\front$\}$\;
       \item \label{step:update_sets}Add $\vec{x}_{i_{min},j_{min}}$ to \known and take it out of \front:\;
    $\quad\quad$\front$\leftarrow$\front $\;\setminus\{\vec{x}_{i_{min},j_{min}}\}\;;$ \known$\leftarrow$\known $\;\cup\{\vec{x}_{i_{min},j_{min}}\} $.\;
    \item \label{step:get_neighborhood}Add the unknown neighborhood of $\vec{x}_{i_{min},j_{min}}$ to \front:\;
    $\quad\quad\cN_{min} = \{\vec{x}_{i_{min}-1,j_{min}},\vec{x}_{i_{min}+1,j_{min}},\vec{x}_{i_{min},j_{min}-1},\vec{x}_{i_{min},j_{min}+1}\}\setminus$\known\;
    $\quad\quad$\front$\leftarrow$\front $\;\cup\; \cN_{min}$.\;
    \item \label{step:iterate_neighbors}\textbf{Foreach }$\vec{x}_{ij}\in\cN_{min}$\;
    $\quad\quad$Update $\tau_{ij}$ by solving the quadratic \eqref{eq:Godunov1}, using only entries in \known. \;
    \textbf{End}

    \end{enumerate}
}
\end{algorithm}

In \cite{sethian1996fast} it was proved that Algorithm \ref{alg:FM} produces a viable viscosity solution to \eqref{eq:Godunov1} when using first order approximations for the first derivatives. Furthermore, it is proved that the values of $\tau$ in the order of which the points are set to \known in Step \ref{step:update_sets} are monotonically increasing.

\subsection{Efficient implementation using minimum heap}

Algorithm \eqref{alg:FM} has two main computational bottlenecks in Steps \ref{step:find_minimum} and \ref{step:iterate_neighbors}, which are repeated $n$ times. For a $d$-dimensional problem, the set \front contains a $d$-1 dimensional manifold of points, of size ${\cal O}(n^{\frac{d-1}{d}})$. To find the minimum of \front efficiently, a minimum heap data structure is used \cite{sethian1996fast,sethian1999fast}.
A minimum heap is a binary tree with a property that a value at any node is less than or equal to the values at its two children. Consequently, the root of the tree holds the minimal value. The simplest implementation of such a tree is done by a sequential array of nodes, using the rule that if a node is located at entry $k$, then its two children are located at entries $2k$ and $2k+1$ (the first element of the array is indexed by $1$, and is the root of the tree). Equivalently, the parent of a node at entry $k$ is located at entry $\lfloor k/2\rfloor$.
Figure \ref{fig:heap} shows an example of a minimum heap and its implementation using array. Generally, each element in the array can hold many properties, and one of these has to be defined as a comparable ``key'', which is used in the heap for sorting. In our case, each node holds a point $\vec{x}$ in the mesh, and its value $\tau(\vec{x})$ as a key.

\begin{figure}
   \centering
   \includegraphics[width=0.4\textwidth]{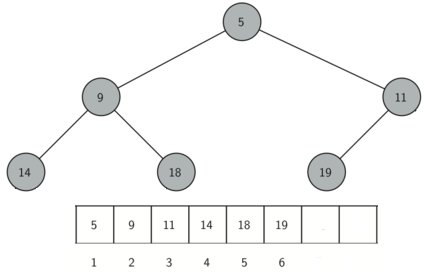}
   \caption{A minimum heap and its implementation using array.}
    \label{fig:heap}
\end{figure}

In its simplest form, the minimum heap structure supports two basic operations: {\tt insert(element,key)} and {\tt getMin()}. For example, this is the case in the C++ standard library implementation of the minimum heap structure. To apply {\tt getMin()}, we remove the first element in the array, and take the last element in the array and push it in the first place. Then, to maintain the property of the heap, we repeatedly replace this value with its smaller child (the smaller of the two) until it reaches down in the tree to the point where it is smaller than its two children or it has no children. The {\tt insert(element,key)} operation first places a new element at the next empty space of the array. Then, it propagates this element upwards, each time replacing it with its parent until it reaches a point where the element is either the root or its key is larger than its parent's key.
Both of the described operations are performed in ${\cal O}(\log m)$ complexity where $m$ is the number of elements in the heap.

In Algorithm \ref{alg:FM}, the set \front is implemented using a min-heap. Steps \ref{step:find_minimum}-\ref{step:update_sets} are trivially implemented using {\tt getMin()}, and {\tt insert(element,key)} is used in Step \ref{step:get_neighborhood}. However, Step \ref{step:iterate_neighbors} of Algorithm \ref{alg:FM} requires updating values which are inside the heap but are not at the root. This operation is not supported in the standard definition of either priority queue or minimum heap. Indeed, the papers \cite{sethian1996fast,sethian1999fast} use a variant of a priority queue, which includes back-links from points $\vec{x}_{ij}$ to their corresponding locations inside the heap, and a more ``software-engineering friendly'' implementation of this idea is suggested in \cite{IMM2001-0841}, where those back-links are incorporated within the heap implementation, without any relation to the mesh. However, although this way an update of a key inside the heap can still be implemented in ${\cal O}(\log m)$, it requires a specialized implementation of the heap, and encumbers the operations described earlier to maintain these back-links. In our implementation, we bypass this need for back-links, and implement Step \ref{step:iterate_neighbors} by reinserting elements to the heap if they are indeed smaller than their value in the heap. In Step \ref{step:find_minimum}, we simply ignore entries which are in \known already. If the algorithm is indeed monotone, like the first order version in \cite{sethian1996fast}, this implementation detail will not change the result of the algorithm. The downside of this change is that it enables \front to grow more than in the back-linked version. However, even if it grows four times compared to the back-linked version, then the heap tree is just two nodes higher, making the difference in running time insignificant.

\subsection{Second order Fast Marching}

Solving the eikonal equation based on a first order discretization in \eqref{eq:Godunov1} provides guaranteed monotonicity and stability.  However, it also provides a less accurate solution because of the added viscosity that is associated with the first order approximation. To get a more accurate FM method, \cite{sethian1999fast} suggests to use a second order upwind approximation in \eqref{eq:Godunov1}, e.g.
\begin{equation}\label{eq:second_order}
D_{i}^{-}\tau = \frac{3\tau_{i} - 4\tau_{i-1} + \tau_{i-2}}{2h},\quad D_{i}^{+}\tau = \frac{-3\tau_{i} + 4\tau_{i+1} - \tau_{i+2}}{2h}.
\end{equation}

However, in some cases, the scheme may revert to first order approximations from certain directions. The obvious case for that is when there are not enough \known points for the high order stencil. This case occurs for example when the given initial region contains only one point. Another condition for using second order operators is given in \cite{sethian1999fast}:
\begin{equation}\label{eq:condition2nd}
\tau_{i-1} \geq \tau_{i-2} \mbox{ or } \tau_{i+1} > \tau_{i+2},
\end{equation}
where the left condition is used for backward operators and the right one for forward operators. If \eqref{eq:condition2nd} is not satisfied, the algorithm reverts to first order operators. Later we show that this condition guarantees the monotonicity of the non-factored FM solution using second order scheme.

%\begin{figure}
%   \centering
%   \includegraphics[width=0.25\textwidth]{bad_second_approx.png}
%   \caption{Negative first derivative high order approximation for a monotone function.}
%    \label{fig:highOrder}
% \end{figure}

%Another condition, suggested by \cite{sethian1999fast}, is to use the high order scheme \eqref{eq:second_order} only if the $\tau$ function is monotone: $\tau_{i-1} > \tau_{i-2}$, or similarly, from the other side, if $\tau_{i+1} > \tau_{i+2}$. We note that such a condition may still result in a situation where using the high order operator may be inappropriate. For example, Figure \ref{fig:highOrder} shows a quadratic interpolation of a monotone function using its value in the equally spaced sample points 1,2,3. The derivative of the interpolated function at the point 3 is equivalent to the result of the approximation \eqref{eq:second_order}, and in this example it comes out negative, while the forward first order approximation comes out positive. In such a case, the corresponding term will be dropped from \eqref{eq:Godunov1}, while a first order term would probably remain and contribute to the solution. We elaborate on this matter in the following sections.

\section{Fast Marching for the factored eikonal equation}

Let us rewrite the factored eikonal equation \eqref{eq:factoredeikonal1} in a squared form, which is closer \eqref{eq:eikonal}:
\begin{equation}\label{eq:factoredeikonal2}
|\tau_0\nabla\tau_1 + \tau_1\nabla\tau_0|^2 = \kappa(\vec{x})^2.
\end{equation}
This writing is the key for deriving the FM algorithm for \eqref{eq:factoredeikonal1}. Similarly to the Godunov upwind scheme in \eqref{eq:Godunov1}, we discretize \eqref{eq:factoredeikonal2} for $\tau_1$ using a derivative operator $\hat{D}$ instead of $D$
\begin{equation}\label{eq:Godunov2}
 \left[\max\{\hat{D}^{-x}_{ij}\tau_1,-\hat{D}^{+x}_{ij}\tau_1,0\}^2 + \max\{\hat{D}^{-y}_{ij}\tau_1,-\hat{D}^{+y}_{ij}\tau_1,0\}^2\right]=\kappa(\vec{x}_{ij})^2,\quad \vec{x}_{ij}\in\Omega_h.
\end{equation}
For example, the backward first order factored derivative operator is given by
\begin{equation}\label{eq:FOfactoredD}
\hat{D}^{-x}_{ij}\tau_1 = (\tau_0)_{ij}\frac{(\tau_1)_{i,j}-(\tau_1)_{i-1,j}}{h} + (p_0)_{ij}(\tau_1)_{ij},
\end{equation}
where $\tau_0$ and $p_0 = \frac{\partial\tau_0}{\partial x}$ are known. From this point we apply the Algorithm \ref{alg:FM} as it is. We hold the values of $\tau_0\tau_1$ in \front, and in Step \ref{step:iterate_neighbors} we update $(\tau_1)_{ij}$ with the solution of \eqref{eq:Godunov2}.

\textbf{Initialization:}
For the non factored equation, Algorithm \ref{alg:FM} is initialized by $\tau(\vec{x}_0) = 0$ at the point source. In the factored equation, this is trivially fulfilled by definition, because at the source $\tau_0(\vec{x}_0) = 0$. Still, $\tau_1(\vec{x}_0)$ should not be chosen arbitrarily since its value is used in the finite difference approximations when evaluating its neighbors. Examining \eqref{eq:factoredeikonal2} at the source yields $\tau_1(\vec{x}_0)^2|\nabla\tau_0|^2 = \tau_1(\vec{x}_0)^2=\kappa(\vec{x}_0)^2$, since we choose $\tau_0$ such that $|\nabla\tau_0|^2 = 1$, independently of $\kappa$. In some cases in the literature, i.e., \cite{fomel2009fast}, the value $\kappa(\vec{x}_0)$ is absorbed in $\tau_0$, such that $|\nabla\tau_0|^2 = \kappa(\vec{x}_0)^2$. Then $\tau_1(\vec{x}_0)$ should be chosen as 1. This is obviously equivalent for computing $\tau$, however, it is much more convenient to choose $\tau_0$ independently of $\kappa$ if one wants to obtain the sensitivities of the FM algorithm (for more details, see Section \ref{sec:sensitivities}).

\textbf{Second order discretization:} Similarly to the non-factored equation, the second order upwind approximations \eqref{eq:second_order} can be used in \eqref{eq:Godunov2}-\eqref{eq:FOfactoredD} for $\tau_1$. Again, we revert to the first order approximation in cases where the additional point needed for the second order approximation is not in \known. We note that unlike the non-factored case, the solution $\tau_1$ is in most cases very smooth at the source (expected to be close to constant or linear). So, when we initialize the algorithm with the value of $\tau_1$ at the point source and revert to a first order approximation for the neighbors, we do not introduce large discretization errors. In the non-factored case, the second derivative of $\tau$ is singular at the source, so using first order approximation there significantly pollutes the rest of the solution.

\subsection{Solution of the piecewise quadratic equation}
We now describe how to solve both the non-factored and factored piecewise quadratic equations \eqref{eq:Godunov1} and \eqref{eq:Godunov2} respectively. This is required in Step \ref{step:iterate_neighbors} of Algorithm \ref{alg:FM}. Solving such an equation consists of the following four steps:
\begin{enumerate}
\item \label{step:order}
Determine the order of approximation for each derivative in \eqref{eq:Godunov1}/\eqref{eq:Godunov2} (Only required for high order schemes).
\item\label{step:direction}
Determine which directions to choose (backward or forward) for each dimension ($x,y$ or $z$).
\item\label{step:solve_all}
Solve the quadratic equation in \eqref{eq:Godunov1}/\eqref{eq:Godunov2}, assuming all terms are positive.
\item\label{step:validity}
Make sure that the solution is valid, such that all $\max$ terms in \eqref{eq:Godunov1}/\eqref{eq:Godunov2} are indeed held with positive values. If not, some terms should be dropped, and the quadratic problem with the remaining terms is solved again.
\end{enumerate}
Let us first consider solving the non-factored first order \eqref{eq:Godunov1}, for which the Step \ref{step:order} is not relevant. In this case, Step \ref{step:direction} is simple: for each $\max\{D^{-}_{ij}\tau,-D^{+}_{ij}\tau,0\}$ term, the smaller of the two values of $\tau$ from both sides (forward or backward) is guaranteed to give a higher finite difference derivative. Furthermore, in Step \ref{step:validity}, if some of the terms turn out negative after Step \ref{step:solve_all}, then we can drop terms from \eqref{eq:Godunov1} in decreasing order of the $\tau$ values, until a valid solution is reached. The same is true for a first order factored version in \eqref{eq:Godunov2}.

However, using second order schemes (selectively) imposes additional complications on the solution of the piecewise quadratic equations \eqref{eq:Godunov1} and \eqref{eq:Godunov2}. There are many options for order of accuracy vs directions in Steps \ref{step:order}-\ref{step:direction}, and in addition it is not clear in which order to drop terms in Step \ref{step:validity} if negative terms are detected. Obviously, one can check all possibilities, but such an option may be costly in high dimensions. To simplify this we follow \cite{sethian1999fast}, and in Step \ref{step:order} we use the second order approximation if the extra point is available in \known and fulfils the condition \eqref{eq:condition2nd}, and revert to first order approximation if not. Then, in Step \eqref{step:direction} we determine the choice of directions considering the non-factored first order approximation \eqref{eq:Godunov1}. That is, if $(\tau_0\tau_1)_{i-1} < (\tau_0\tau_1)_{i+1}$, then we choose the backward upwind direction; otherwise we choose the forward direction. That is done correspondingly for each dimension.

Once Steps \ref{step:order}-\ref{step:direction} are done, \eqref{eq:Godunov2} reduces to a piecewise quadratic equation of the form
\begin{equation}\label{eq:piecewiseQuad}
\sum_{k}{\max\{\alpha_k(\tau_{ij} - \beta_k),0\}}^2 = \kappa(\vec{x}_{ij})^2,
\end{equation}
where $\alpha_k \geq 0$, $\beta_k \geq 0$ are non-negative constants that are coming from the finite difference approximations. For example, assuming that $k=1$ corresponds to the $x$ coordinate, then \eqref{eq:FOfactoredD} would correspond to $\alpha_1 = \frac{(\tau_0)_{i,j}}{h} + (p_0)_{i,j}$ and $\beta_1 = \frac{(\tau_0)_{i,j}(\tau_1)_{i-1,j}}{h\alpha_1}$. In Step \ref{step:solve_all} we simply ignore the $\max\{\cdot,0\}$ function and solve the equation assuming all terms are positive. We solve a simple quadratic function and choose the larger one of its two solutions for $\tau_{ij}$. If all chosen derivative terms are positive, the solution is valid; otherwise, we reduce the terms in \eqref{eq:piecewiseQuad} in decreasing order of $\beta_k$, each time solving \eqref{eq:piecewiseQuad} with the remaining terms until a valid solution is reached. In three dimensions for example, this involves at most three quadratic solves. Algorithm \ref{alg:peicewiseQuad} summarizes the solution of the piecewise quadratic equation.

\begin{algorithm}
\DontPrintSemicolon
\label{alg:peicewiseQuad}
\KwSty{Algorithm: Solution of the piecewise quadratic equation}\;
\For{each dimension x,y,...}{
\textit{\% Choosing direction, forward or backward.}\;
\uIf {both forward and backward neighboring points are in \known} {Choose the direction with smaller neighboring $\tau$.\;}
\Else{Otherwise, choose the direction in \known.}
\textit{\% Choosing order of approximation, 1st or 2nd.}\;
\uIf {next neighboring point is in \known}{Use second order approximation.\;}
\Else{Use first order approximation.}
}
\textit{\% Now all coefficients of $\alpha_k$ and $\beta_k$ of Equation \eqref{eq:piecewiseQuad} are known.}\;
Calculate $(\tau_1)_{i,j}$ by solving Equation \eqref{eq:piecewiseQuad}.\;
\While{the solution $(\tau_1)_{i,j}$ is not valid}{
    Remove the term with largest $\beta_k$ from the remaining terms in \eqref{eq:piecewiseQuad}.\;
    Calculate $(\tau_1)_{i,j}$ by solving \eqref{eq:piecewiseQuad} with the remaining terms.\;

}

\end{algorithm}

\subsection{The monotonicity of the obtained solution}
It is known that the solution of \eqref{eq:eikonal} is monotone in the direction of the characteristics. We now show how to enforce the monotonicity of our solution using the FM method for the factored eikonal equation. To set the stage, we first consider the FM method for the original non-factored equation.

In \cite{sethian1996fast} the non-factored first order discretization \eqref{eq:Godunov1} is considered. In this case, each newly calculated value $\tau_{ij}$ is guaranteed to be larger than its \known neighbors at the time of the calculation. To show this clearly, consider for example that the backward derivative is chosen in the $x$ direction. Then,
\begin{equation}\label{eq:mono_non_factored1}
\tau_{i-1,j} = \tau_{i,j} - h\frac{\tau_{i,j} - \tau_{i-1,j}}{h} = \tau_{i,j} - hD^{-x}_{i,j}\tau,
\end{equation}
so since $D^{-x}_{i,j}\tau \geq 0$ in the solution of \eqref{eq:Godunov1}, we have that $\tau_{i,j} \geq \tau_{i-1,j}$. This means that $\tau_{i,j}$ is greater or equal to its \known neighbors. This property insures the monotonicity of the solution. The proof for this appears in \cite{sethian1996fast}, but here we can simplify it because unlike \cite{sethian1996fast}, we only calculate entries using \known values.
We state the following lemma:

\begin{lemma} \label{lemma:monotonicity}
Let $\tau$ be the result of Algorithm \ref{alg:FM}, for solving the (non-factored) first order equation \eqref{eq:Godunov1}. Then the values of $\tau$ are monotonically non-decreasing in the order in which they are set to \known.
\end{lemma}

\begin{proof}
Denote by $\vec{x}_k$ an element that is set to \known at Step \ref{step:update_sets} of the $k$-th iteration of Algorithm \ref{alg:FM}. Assume by contradiction that there exists two elements $\vec{x}_p$ and $\vec{x}_k$, such that $\tau(\vec{x}_p) > \tau(\vec{x}_k)$ and $p < k$. Without loss of generality, assume that $k$ is the earliest iteration that this condition is fulfilled. Let $\bar{k}<k$ be the iteration in which the value of $\tau(\vec{x}_k)$ is updated in the last time and it is entered to \front. We know that $\vec{x}_{\bar{k}}$ is a neighbor of $\vec{x}_{k}$. By the algorithm, we know that at the $\bar{k}$-th iteration $\vec{x}_p$ is already set to \known, otherwise $\vec{x}_k$ would have been chosen to \known at the $p$-th iteration instead of $\vec{x}_p$. By the assumption, we know that $ \tau(\vec{x}_p)\leq \tau(\vec{x}_{\bar{k}})$, because otherwise $\vec{x}_k$ would not have been the earliest element to violate the monotonicity. By the property in \eqref{eq:mono_non_factored1}, we know that $\tau(\vec{x}_{\bar{k}}) \leq \tau(\vec{x}_{k})$, and hence we reach $ \tau(\vec{x}_p)\leq \tau(\vec{x}_{k})$, which contradicts our assumption.
\end{proof}

Furthermore, the lemma above can be extended for a Fast Marching solution of \emph{any} equation \eqref{eq:Godunov1} such that the discretization operator $D$ satisfies a monotonicity condition:
\begin{equation}\label{eq:mono_discretization}
D_{ij}^{-x}\tau \geq 0 \Rightarrow \tau_{ij} \geq \tau_{i-1,j} \mbox{ and } -D_{ij}^{+x}\tau \geq 0 \Rightarrow \tau_{ij} \geq \tau_{i+1,j}.
\end{equation}
\noindent The next corollary can be proved using the same arguments as in Lemma \ref{lemma:monotonicity}:
\begin{corollary} \label{corr:monotonicity}
Let $\tau$ be the result of Algorithm \ref{alg:FM}, for solving the Godunov upwind equation \eqref{eq:Godunov1} using operators $D$ which satisfy \eqref{eq:mono_discretization}. Then the values of $\tau$ are monotonically non-decreasing in the order in which they are set to \known.
\end{corollary}

The condition \eqref{eq:mono_discretization} and the corollary above is violated when the second order operators \eqref{eq:second_order} are used in \eqref{eq:Godunov1}. However, if we look at a single violation, then it is of order $h^2$. To show this, we examine the backward difference derivative using Taylor expansion:
\begin{equation}\label{eq:mono_non_factored2}
\tau_{i-1,j} = \tau_{i,j} - h\left(\frac{\partial\tau}{\partial x}\right)_{i,j} + O(h^2) = \tau_{i,j} - hD^{-x}_{i,j}\tau + O(h^2),
\end{equation}
where $D^{-x}_{i,j}$ is given in \eqref{eq:second_order} (the same arguments can be derived for the forward difference derivative). Assuming again that $D^{-x}_{i,j}\tau>0$, this means that each newly calculated $\tau_{i,j}$ is generally greater than its \known neighbors, but may violate that up to magnitude $O(h^2)$. Note that if $D^{-x}_{i,j}\tau$ is sufficiently bounded away from zero and the second derivative $\frac{\partial^2\tau}{\partial x^2}$ is bounded in $[x_{i-1},x_i]$, then $\tau_{i,j} > \tau_{i-1,j}$ will be satisfied.

To correct this and obtain a monotone solution using \eqref{eq:second_order}, one may impose the condition \eqref{eq:condition2nd} for using the second order scheme. If it is not satisfied, the scheme reverts to the first order scheme, which satisfies \eqref{eq:mono_discretization}. If \eqref{eq:condition2nd} is satisfied, then \eqref{eq:second_order} does satisfy \eqref{eq:mono_discretization}, because
\begin{equation}
2hD^{-x}_{ij}\tau = 3\tau_{ij} - 4\tau_{i-1,j} + \tau_{i-2,j} < 3\tau_{ij} - 3\tau_{i-1,j}.
\end{equation}
Note that the condition \eqref{eq:condition2nd} is suggested in \cite{sethian1999fast} but the monotonicity guarantee of the second order scheme is not examined.

We now examine the monotonicity of the obtained factored solution $\tau_0\tau_1$ when using first order operators in \eqref{eq:Godunov2}. Suppose that we are calculating $(\tau_1)_{ij}$ using the backward operator \eqref{eq:FOfactoredD} in the $x$ direction in \eqref{eq:Godunov2}. We again start with a Taylor expansion
\begin{equation}
\begin{array}{rcl}
(\tau_0\tau_1)_{i-1,j} &=& \left((\tau_0)_{ij} - h\left(\frac{\partial\tau_0}{\partial x}\right)_{i,j} + O(h^2)\right)\left((\tau_1)_{ij} - h\left(\frac{\partial\tau_1}{\partial x}\right)_{i,j} + O(h^2)\right) \\
&=& (\tau_0\tau_1)_{ij} - h(\tau_0)_{ij}\left(\frac{\partial\tau_1}{\partial x}\right)_{i,j} - h(\tau_1)_{ij}\left(\frac{\partial\tau_0}{\partial x}\right)_{i,j} + O(h^2)\\
                         &=& (\tau_0\tau_1)_{ij} - h\hat{D}^{-x}_{ij}\tau_1 + O(h^2),
\end{array}
\end{equation}\label{eq:Taylor1}
\noindent where the last equality is obtained by placing $\left(\frac{\partial\tau_1}{\partial x}\right)_{i,j} = \frac{(\tau_1)_{i,j}-(\tau_1)_{i-1,j}}{h} + O(h)$.
This expansion shows that if the monotonicity is not obtained, i.e., $(\tau_0\tau_1)_{i,j} - (\tau_0\tau_1)_{i-1,j}<0$, then the non-factored derivative is negative, $D^{-x}_{ij}(\tau_0\tau_1) < 0$, while the factored derivative is non-negative $\hat{D}^{-x}_{ij}\tau_1 \geq 0$ (otherwise it is not chosen in \eqref{eq:Godunov2}).  This means that the monotonicity may be violated only up to an error of $O(h^2)$. This holds for both first and second order upwind approximations. In fact, \eqref{eq:Taylor1} shows that this is a result of using the chain rule rather than the order of discretization of the operators $\hat D$, since the monotonicity condition involves the value $(\tau_0)_{i-1,j}$, while it does not appear in the discretization scheme. In any case, the magnitude of the error in the monotonicity violation is either of the same or of higher order as the error in $\tau_1$, using first or second order schemes. Again, if $D^{-x}_{i,j}\tau_1$ is sufficiently bounded away from zero and $\frac{\partial^2\tau_1}{\partial x^2}$ is bounded in $[x_{i-1},x_i]$, then the monotonicity $(\tau_0\tau_1)_{i,j} > (\tau_0\tau_1)_{i-1,j}$ will be satisfied.

Nevertheless, in our algorithm we may enforce the monotonicity of the obtained solution by reverting to the non-factored operators in cases where the monotonicity is not satisfied, or, the factored and non-factored schemes do not agree in sign, for example: $\hat{D}^{-x}_{ij}\tau_1 \geq 0$, but $D^{-x}_{ij}(\tau_0\tau_1) < 0$. Note that in this case the numerical derivative is approximately zero, hence the direction of the characteristic is almost parallel to the $y$ direction. We apply this change using the same order of derivative which the algorithm chooses to use. That is, if the algorithm chooses a first or second order factored stencil, we revert to a standard first or second order stencil, respectively. Following Corollary \ref{corr:monotonicity}, this guarantees the monotonicity of the solution, because we enforce the condition \eqref{eq:mono_discretization} at all stages of the algorithm. We note that experimentally, this small correction does not influence the accuracy of the solution obtained with our algorithm in both first and second order schemes in two and three dimensions.

\section{Calculation of sensitivities and travel time tomography}\label{sec:sensitivities}
Travel time tomography is a useful tool in some Geophysical applications. One way to obtain it is by using the eikonal equation as a forward problem inside an inverse problem \cite{sei1994gradient}. To solve the inverse problem, one should be able to solve \eqref{eq:eikonal} accurately, and to compute its sensitivities. The works of \cite{leung2006adjoint,taillandier2009first} computes the tomography by FS, and require an FS iterative solution for computing the sensitivities. When using the FM algorithm for forward modeling, those are obtained more efficiently by a simple solution of a lower triangular linear system \cite{li2013first,benaichouche2015first}. More explicitly, let us denote by boldface all the discretized values of the mentioned functions on a grid, and suppose that we set $\bfm$ to be the vector of the values of $\kappa(\vec{x})^2$ on this grid. By solving \eqref{eq:Godunov2}, we get a function $\bftau_1(\bfm)$ for the values of $\tau_1$ on the grid. We wish to get a linearization for $\bftau_1(\bfm)$, such that we can predict its change following a small change in $\bfm$. That is, we wish to be able to apply an approximation
\begin{equation}
\bftau_1(\bfm+\bfdelta\bfm) \approx \bftau_1(\bfm) + \bfJ\bfdelta\bfm,
\label{eq:Sensitivity}
\end{equation}
where $\bfJ$ is the sensitivity matrix (or Jacobian) defined by
\begin{equation}
\bfJ_{ij} = (\grad_{\bfm} \bftau_1)_{ij} =  {\frac {\partial (\bftau_1)_{i}}{\partial \bfm_{j}}}.
\end{equation}

To obtain the sensitivity we first rewrite \eqref{eq:Godunov2} in implicit form
\begin{equation}\label{eq:Godunov2algebraic}
\bff(\bfm,\bftau_1) = (\hat{\bfD}^{x}\bftau_1)^2 + (\hat{\bfD}^{y}\bftau_1)^2 - \bfm = 0,
\end{equation}
where $\hat{\bfD}^{x} = \diag(\bftau_0)\bfD^{x} + \diag(\bfp_0)$ and $\hat{\bfD}^{y} = \diag(\bftau_0)\cdot\bfD^{y} + \diag(\bfq_0)$ are the matrices that apply the finite difference derivatives that are chosen by the FM algorithm when applied for $\bfm$. $\bfp_0$ and $\bfq_0$ are the analytical derivatives of $\tau_0$ with respect to $x$ and $y$ on the grid respectively, and $\diag(\bfx)$ denotes a diagonal matrix whose diagonal elements are those of the vector $\bfx$. We note that in the points where no derivative is chosen in the solution of \eqref{eq:Godunov2}, a zero row is set in the corresponding operator $\hat{\bfD}$. Also, at the row of the point source, we set each of $\hat\bfD^{x}$ and $\hat\bfD^{y}$ to have only one diagonal non-zero element, which equals to the values of $\bfp_0$ and $\bfq_0$ at the source. This way, \eqref{eq:Godunov2algebraic} is exactly fulfilled for $\hat\bfD^{x}$ and $\hat\bfD^{y}$ and $\bftau_1$.

To obtain the sensitivity, we apply the gradient operator to both sides of \eqref{eq:Godunov2algebraic}, yielding  $(\nabla_{\bftau_1}\bff)(\nabla_\bfm\bftau_1) + \nabla_\bfm\bff = 0$, and define \cite{haber2014computational}:
\begin{eqnarray}\label{eq:Sensitivity1}
\bfJ(\bfm) = \grad_{\bfm} \bftau_1 =  -(\grad_{\bftau_1}\bff)^{-1} (\grad_{\bfm} \bff).
\end{eqnarray}
This results in
\begin{eqnarray}\label{eq:Sensitivity2}
\bfJ = (\diag(2\hat{\bfD}^{x}\bftau_1)\hat{\bfD}^{x} + \diag(2\hat{\bfD}^{y}\bftau_1)\hat{\bfD}^{y})^{-1},
\end{eqnarray}
following $\grad_{\bfm} \bff = -\bfI$, and since the operators $\hat\bfD$ do not depend on $\bfm$ (we defined $\tau_0$ and its derivatives so it does not depend on $\kappa$).

The matrix \eqref{eq:Sensitivity2} can be multiplied with any vector efficiently given the order of variables $(i,j)$ in which the FM algorithm set their values as \known. To apply $\bfJ$ on an arbitrary vector $\bfx$, i.e. calculate $\bfe=\bfJ\bfx$, a linear system $\bfA\bfe=\bfx$ can be solved with $\bfA = \bfJ^{-1}$ (note that $\bfA$ is a sparse matrix). The equations of this linear system, which correspond to the rows of $\bfA$, can be approached and solved sequentially in the FM order of variables. Since the FM algorithm uses only \known variables for determining each new variable, then when looking at each row $i$ of $\bfA$, the non-zero entries in that row (except $i$) correspond to variables that where in \known when $\bftau_i$ was determined during the FM run. Therefore, if all those variables are known except $i$, then the $i$-th equation has only one unknown ($e_i$) and can be trivially solved. In other words, if we permute $\bfA$ according to the FM order, we get a sparse lower triangular matrix, and the corresponding system can be solved efficiently in one forward substitution sweep in $O(n)$ operations. For the non-factored equation one may use \eqref{eq:Sensitivity2} with non-factored operators $\bfD^{x}$ and $\bfD^{y}$ instead of the factored ones \cite{li2013first}.

\subsection{Travel time tomography using Gauss-Newton}

Assume that we have several sources and receivers set on an open surface, and for each source we have traveltime data $\bfd_{\obs}^{i}$ given in the location of the receivers. Based on these observations we wish to compute the unknown slowness model of the ground underneath. The inverse problem for this process, called travel time tomography, may be given by
\begin{eqnarray}
\label{eq:inverseEik}
\min_{m_{L}<\bfm<m_{H}}\phi(\bfm) = \min_{m_{L}<\bfm<m_{H}}\left\{\sum_{i=1}^{n_s}{\|\bfP^{\top} \bftau^{i}(\bfm) - \bfd_{\obs}^{i}\|^2} + \alpha R(\bfm)\right\},
\end{eqnarray}
where
\begin{eqnarray}
 \label{eq:eik_i}
 |\grad \bftau^{i}|^{2} = \bfm(\vec x)\quad \bftau^{i}(\vec x_{i}) = 0\quad i = 1,\ldots,n_{s}
\end{eqnarray}
Here $\bftau^{i}$ is the travel time from the point source $\vec x_i$, and $\bfm(\vec x)=\kappa(\vec x)^2$ is the squared slowness model
as in \eqref{eq:eikonal}, only now it is unknown. The operator $\bfP^{\top}$ is a projection to the set of receivers that gather the wave information. Here we assume that the information from all sources is available on all the receivers, i.e., the projection operator $\bfP$ does not change between sources. $R(\bfm)$ is a regularization term and $\alpha>0$ is its balancing parameter. The parameters $m_{L}$ and $m_{H}$ are positive lower and upper bounds needed for keeping the slowness of the medium physical. We note that the observations $\bfd_{\obs}^{i}$ can be obtained manually from recorded seismic data or by automatic time picking---for more information see \cite{saragiotis2013automatic} and references therein.

Without the regularization term $R(\bfm)$, the problem \eqref{eq:inverseEik} is ill-posed, i.e., many solutions $\bfm$ may fit the predicted travel time to the measured data \cite{vogel2002computational,somersalo2004statistical}. For this reason, in most cases we cannot expect to exactly recover the true model, but wish to recover a reasonable model by adding prior information using the regularization term $R(\bfm)$. This term aims to promote physical or meaningful solutions that we may expect to see in the recovered model. For example, in seismic exploration, one may expect to recover a layered model of the earth subsurface, hence may choose $R$ to promote smooth or piecewise-smooth functions like the total variation regularization term \cite{rudin1992nonlinear}.

There are several ways to solve \eqref{eq:inverseEik}, and most of them are gradient-based. Here we focus on Gauss-Newton. This method is computationally favorable here, since its cost is governed by the application of sensitivities, which are easy to obtain using FM. Given an approximation $\bfm^{(k)}$ at the $k$-th iteration, we place \eqref{eq:Sensitivity} into \eqref{eq:inverseEik} and get
\begin{eqnarray}
\label{eq:GNapprox}
\min_{\delta\bfm}\ \hf \sum_{i=1}^{n_{s}} \| \bfP^{\top} \diag(\bftau_0^i)\left(\bftau_1^i(\bfm^{(k)}) + \bfJ^i\delta\bfm\right) -  \bfd_{\obs}^{i} \|^{2} + \alpha R(\bfm^{(k)}+\delta\bfm),
\end{eqnarray}
where $\bfJ^i$ is the sensitivity of $\bftau^i_1$ at $\bfm^{(k)}$.
Minimizing this approximation for $\delta\bfm$ leads to computing the gradient
\begin{eqnarray}
\label{gradpi}
\grad_{\bfm}\phi(\bfm^{(k)}) = \sum_{i=1}^{n_{s}} (\bfJ^i)^{\top}\diag(\bftau_0^i)\bfP\left( \bfP^{\top}\diag(\bftau_0^i) \bftau_1^i -  \bfd_{\obs}^{i}\right) + \alpha \grad_{\bfm} R(\bfm^{(k)}).
\end{eqnarray}
We then approximately solve the linear system
$$ \bfH \delta \bfm = -\grad_{\bfm}\phi(\bfm^{(k)}) $$
where
$$
\bfH =  \sum_{i=1}^{n} (\bfJ^i)^{\top}\diag(\bftau_0^i) \bfP \bfP^{\top} \diag(\bftau_0^i)\bfJ^i + \alpha \Delta_\bfm R(\bfm^{(k)}).
$$
The linear system is solved using the conjugate gradient method where only matrix vector products are computed.
Finally, the model is updated, $\bfm \leftarrow \bfm + \mu \delta \bfm$ where $\mu \le 1$ is a line search parameter that is chosen such that the objective function is decreased at each iteration.

\section{Numerical results: solving the eikonal equation}
In this section we demonstrate the FM algorithm using first or second order upwind discretization for solving the factored eikonal equation \eqref{eq:factoredeikonal1}. We demonstrate both the accuracy of the obtained solution, and the computational cost of calculating it using the FM algorithm. The accuracy of the algorithm is demonstrated by two error norms: one in the maximum norm $l_\infty$, and one is the mean $l_2$ norm defined by the standard $l_2$ norm of the error divided by the square root of the total number of variables. Similarly to \cite{sethian1999fast}, we show these two measures to demonstrate the accuracy of the second order scheme. Showing the $l_\infty$ norm of the error for this scheme may result in only first order accuracy, because at some points our second order FM algorithm reverts to first order operators, which may be picked by the $l_\infty$ norm.

To demonstrate the efficiency of the computation, we measure the time in which the algorithm solves each test. We also show this timing in terms of work-units, where each work unit is defined by the time that it takes to evaluate the equation \eqref{eq:eikonal} using given central difference gradient stencils (without memory allocation time). We note that the more reliable timings appear for the large scale examples.

We use analytical examples for media where there is a known analytical solution for a point source located at $\vec{x}_0$. The first two appear in \cite{fomel2009fast}. We show results for two and three dimensions. Our code is written in Julia language \cite{Julia} version 0.4.5, and all our tests were calculated on a laptop machine using Windows 10 64bit OS, with Intel core-i7 2.8 GHz CPU with 32 GB of RAM. Our code is publicly available in \url{https://github.com/JuliaInv/FactoredEikonalFastMarching.jl}. We do not enforce the monotonicity in the results below, but those can be enforced in our package. The three test cases are listed below.
\paragraph{Test case 1: Constant gradient of squared slowness} In this test case we set:
\begin{equation}
  \kappa^2(\vec{x}) = s_0^2 + 2a \vec{e}_1\cdot (\vec{x}-\vec{x}_0),
\end{equation}
where $\vec{e}_1=(1,0)$ is a unit vector, and $\cdot$ is the standard dot product. The parameters $a$, $s_0$, the domain and the source location are chosen differently in 2D and 3D. The corresponding exact solution is given by
\begin{equation}
\tau_{exact}(\vec{x})= \bar S^2 \sigma - \frac{1}{6}a^2(\sigma^3),
\end{equation}
where
\begin{eqnarray}
  \bar S^2(\vec{x}) &=& s_0^2 + a\vec{e}_1\cdot(\vec{x}-\vec{x}_0)\\
  \sigma^2(\vec{x}) & = & \left(\bar S^2 + \sqrt{\bar S^4 - a^2\|\vec{x}-\vec{x_0}\|^2}\right)^{-1}2\|\vec{x}-\vec{x_0}\|^2.
\end{eqnarray}
Figures \ref{fig:kappa1} and \ref{fig:kappa1C} show the model $\kappa$ for this test case with the chosen parameters for 2D.

%\begin{figure}[h]
%   \centering
%   \label{fig:kappa1}
%   \includegraphics[width=0.6\textwidth]{Kappa1.png}
%   \caption{The 2D slowness model $\kappa(\vec{x})$ of test case 1: constant gradient of squared slowness.}
% \end{figure}

\paragraph{Test case 2: Constant gradient of velocity} In this test case we set:
\begin{equation}
\kappa(\vec{x}) = \left(\frac{1}{s_0} + a\vec{e}_1\cdot(\vec{x}-\vec{x}_0) \right)^{-1},
\end{equation}
where again $\vec{e}_1=(1,0)$, $\cdot$ is the dot product, and the parameters $a$, $s_0$, the domain and the source location are chosen differently in 2D and 3D. The exact solution is given by
\begin{equation}
\tau_{exact}(\vec{x})= \frac{1}{a}\mbox{acosh}\left(1+\frac{1}{2}s_0a^2\kappa(\vec{x})\|\vec{x}-\vec{x}_0\|^2\right).
\end{equation}
Figures \ref{fig:kappa2} and \ref{fig:kappa2C} show the model $\kappa$ for this test case with the chosen parameters for 2D.

 %\begin{figure}[h]
%   \centering
%   \label{fig:kappa2}
%   \includegraphics[width=0.6\textwidth]{Kappa2.png}
%   \caption{The 2D slowness model $\kappa(\vec{x})$ of test case 2: constant gradient of velocity.}
% \end{figure}

\paragraph{Test case 3: Gaussian factor} In this test case we choose a function for $\tau_1^{exact}$ and multiply it by $\tau_0$ to get $\tau_{exact}$. We choose $\tau_1$ as a Gaussian function centered around a point $\vec{x_1}$:
\begin{equation}
  \tau_1^{exact}(\vec{x}) = \frac{1}{2}\exp\left(-(\vec{x}-\vec{x}_1)^T\Sigma(\vec{x}-\vec{x}_1)\right) +\frac{1}{2},
\end{equation}
where $\Sigma$ is a $2\times2$ or $3\times3$ positive diagonal matrix. As before, the parameters $\vec{x}_1$ and $\Sigma$, the domain and the source location $\vec{x_0}$ are chosen differently in 2D and 3D. Here $\kappa(\vec{x})$ is defined by \eqref{eq:factoredeikonal2}, with $\tau_0$ being the distance function.
Figures \ref{fig:kappa3} and \ref{fig:kappa3C} show the model $\kappa$ for this test with the chosen parameters for 2D.

 %\begin{figure}[h]
%   \centering
%   \label{fig:kappa3}
%   \includegraphics[width=0.6\textwidth]{Kappa3.png}
%   \caption{The 2D slowness model $\kappa(\vec{x})$ of test case 3: Gaussian factor.}
% \end{figure}

\begin{figure}
\begin{center}
  \subfigure[Test case 1: Constant gradient of squared slowness.]{\includegraphics[width=0.32\linewidth]{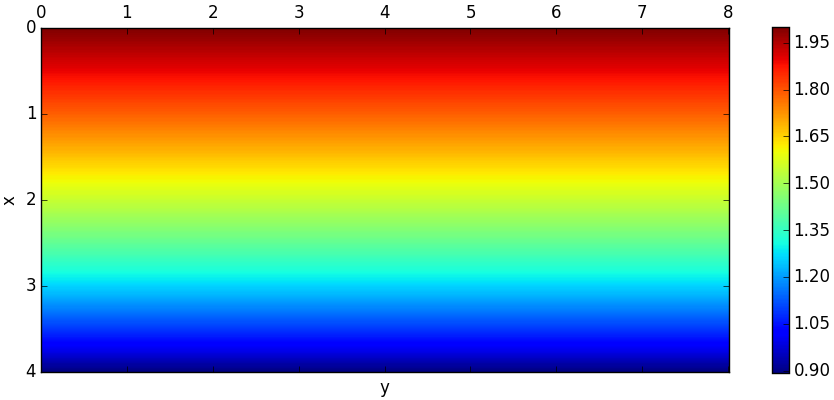}\label{fig:kappa1}}
  \subfigure[Test case 2: Constant gradient of velocity.]{\includegraphics[width=0.32\linewidth]{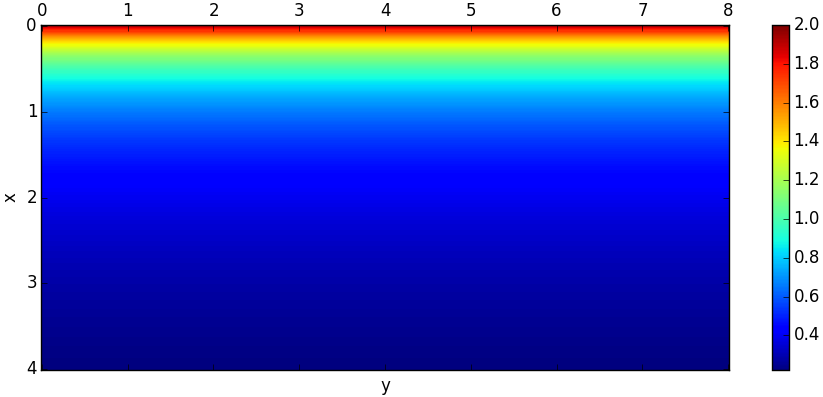}\label{fig:kappa2}}
  \subfigure[Test case 3: Gaussian factor]{\includegraphics[width=0.32\linewidth]{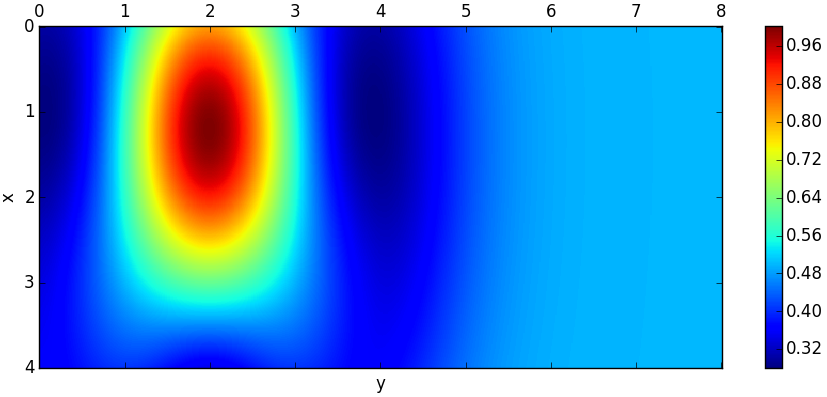}\label{fig:kappa3}}
  \subfigure[Test case 1: Contours of the solution.]{\includegraphics[width=0.32\linewidth]{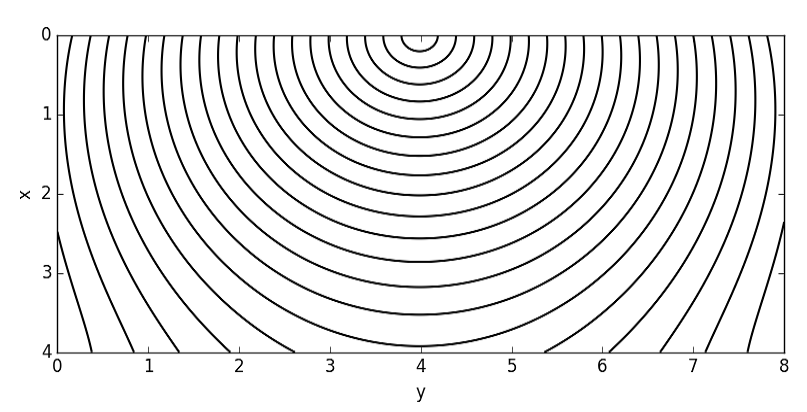}\label{fig:kappa1C}}
  \subfigure[Test case 2: Contours of the solution.]{\includegraphics[width=0.32\linewidth]{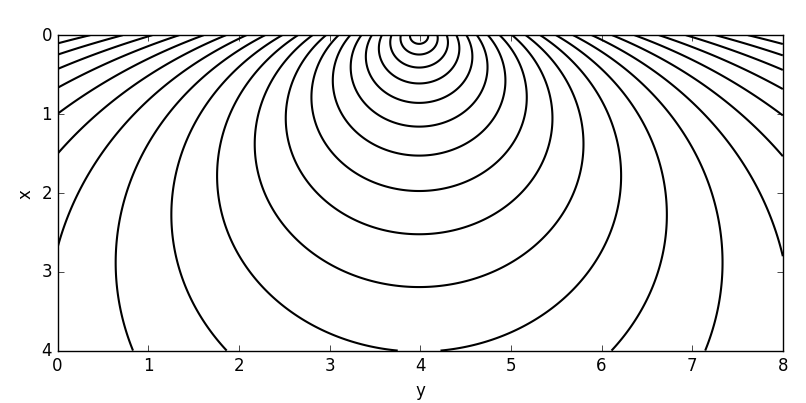}\label{fig:kappa2C}}
  \subfigure[Test case 3: Contours of the solution.]{\includegraphics[width=0.32\linewidth]{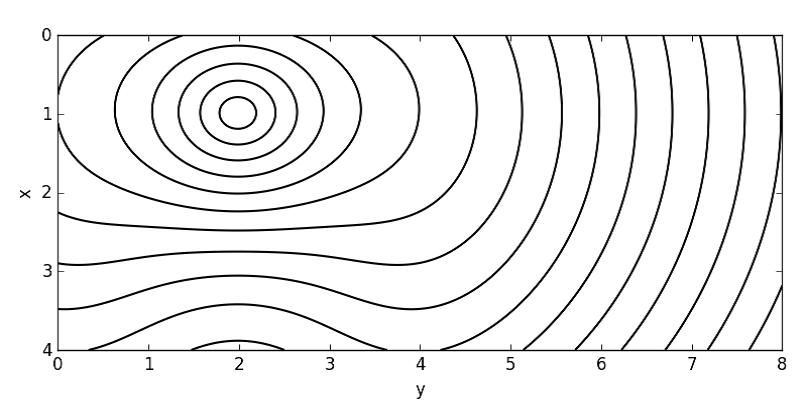}\label{fig:kappa3C}}
\end{center}
\caption{The 2D slowness model $\kappa(\vec{x})$ of the three test cases and the corresponding contours of the 2D solutions.}
\label{fig:TTT}
\end{figure}

\subsection{Two dimensional tests}
Now we show results for the two dimensional versions of the tests mentioned above. For all tests in 2D we choose the domain to be $[0,4] \times [0,8]$, while $h=h_x=h_y$ varies from large to small.

\paragraph{Test case 1: Constant gradient of squared slowness}

For this 2D setting, we use the parameters $a  = -0.4$, $s0 = 2.0$, and the source location is $\vec{x}_0 = (0,4)$.
Table \ref{tab:test1-2D} summarizes the results for this test. On the first order section we see a typical first order convergence rate in both error norms. As the mesh size increases by two in each direction, the errors drop by a factor of two. In the second order section we see the typical behavior of the FM algorithm. At some points, first order operators are used, and hence the error at those locations dominates the $l_\infty$ norm. Still, we observe much better convergence compared to the first order $l_\infty$, only it is not of second order. In the mean $l_2$ norm we see typical second order convergence---as the mesh size increases by two in each direction, the errors drop by a factor of four. In any case, the errors in the second order columns are much smaller than those in the first order columns.

In terms of computational cost, the 2D FM algorithm exhibits favorable timings and work counts. Except the small cases, the cost of the algorithm is comparable to 200-300 function evaluations using standard difference stencils. This is maintained for all the considered mesh sizes. The difference in the computational cost between using first and second order schemes is only about $10\%$ of execution time.

\begin{table}
\centering
\begin{tabular}{|c|c|c|c|c|c|c|c|}
\hline
%\mulcol{2}{|c|}{&$2^{nd}}
$h$& $n$ &\mulcol{2}{|c|}{$1^{st}$ order} &  \mulcol{2}{|c|}{$2^{nd}$ order}\\
\hline
 & &error in $\tau$  & time(work) & error in $\tau$ & time(work)\\
%\hline
%1/20 & $81\times161$ &[7.47e-03, 1.90e-03]& 0.12s(3050) &   [2.45e-04, 3.94e-05] & 0.12s(3030)\\
\hline
1/40 & $161\times321$ &       [3.71e-03, 9.42e-04] &      0.05s(217) &    [9.33e-05, 9.26e-06] &      0.05s (202)\\
\hline
1/80 & $321\times641$  &     [1.85e-03, 4.69e-04]   &     0.19s (199)  &   [3.30e-05, 2.21e-06]    &     0.20s (209)\\
\hline
1/160 & $641\times1281$&     [9.22e-04, 2.34e-04]   &     0.85s (217)  &   [1.14e-05, 5.32e-07]    &     0.85s (218)\\
\hline
1/320& $1281\times2561$  &    [4.60e-04, 1.17e-04]    &    3.89s (266)   & [4.06e-06, 1.28e-07]     &   3.84s (262)  \\
\hline
1/640 &$2561\times5121$&       [2.30e-04, 5.83e-05]  &     16.4s (278)    &   [1.47e-06, 3.12e-08]    &  17.1s (289)    \\
\hline
1/1280 &$5121\times10241$&       [1.15e-04, 2.92e-05]  &     76.6s (316)    &   [5.18e-07, 7.64e-09]    &  77.5s (320) \\
\hline

\end{tabular}
\caption{Results for 2D constant gradient of squared slowness (test case 1). The error measures are in the $[l_\infty, \mbox{mean } l_2]$ norms.}\label{tab:test1-2D}
\end{table}

\paragraph{Test case 2: Constant gradient of velocity}
For this 2D setting, we use the parameters $a  = 1.0$, $s_0 = 2.0$, and the location of the source is again at $\vec{x}_0 = (0,4)$.
Table \ref{tab:test2-2D} summarizes the results for this test. The results here are almost identical to the previous test case. The first order columns show typical first order convergence in both error norms. The second order columns show better convergence and exhibits second order convergence in the mean $l_2$ norm column.
The computational costs columns show timings which are almost identical to the previous test case.

\begin{table}[ht]
\centering
\begin{tabular}{|c|c|c|c|c|c|c|c|}
\hline
%\mulcol{2}{|c|}{&$2^{nd}}
$h$& $n$ &\mulcol{2}{|c|}{$1^{st}$ order} &  \mulcol{2}{|c|}{$2^{nd}$ order}\\
\hline
 & &error in $\tau$  & time(work) & error in $\tau$ & time(work)\\
\hline
1/40 & $161\times321$ &   [2.66e-02, 1.01e-02]  &   0.05s (205)  & [4.86e-04, 2.90e-04]      &      0.05s (236)  \\
\hline
1/80 & $321\times641$  &     [1.32e-02, 5.05e-03]     &   0.21s (221)     &     [1.67e-04, 7.38e-05]   &0.20s (206)   \\
\hline
1/160 & $641\times1281$&     [6.59e-03, 2.52e-03]    &  0.87s (223)     &  [5.18e-05, 1.85e-05]    &  0.86s (221)  \\
\hline
1/320& $1281\times2561$  &      [3.29e-03, 1.26e-03]   &   3.80s (259)  &   [1.90e-05, 4.61e-06]   &  3.88s (265)   \\
\hline
1/640 &$2561\times5121$&       [1.65e-03, 6.28e-04]    &  16.2s (274)    &  [6.58e-06, 1.15e-06]   & 16.6s (280)    \\
\hline
1/1280 &$5121\times10241$&        [8.22e-04, 3.14e-04]  &    73.8s (304)  &    [2.28e-06, 2.86e-07]  & 74.6s (307) \\
\hline
\end{tabular}
\caption{Results for the 2D constant gradient of velocity (test case 2). The error measures are in the $[l_\infty, \mbox{mean } l_2]$ norms.}\label{tab:test2-2D}
\end{table}

\paragraph{Test case 3: Gaussian factor}
For this setting, we use the parameters $\Sigma = \diag(0.1,0.4)$, $\vec{x}_1 = (4/3,2)$ (floored to the closest grid point), and the source is located in the point $\vec{x}_0 = (1,2)$. Table \ref{tab:test3-2D} summarizes the results for this test. Again we see first order convergence at the first order columns in both norms. On the second order columns we again see faster convergence, and in the mean $l_2$ norm column we see convergence rate that is close to second order---the error decreases by a factor of about 3.9 when the mesh size increases by a factor of 2 in each dimension. Again we see similar behavior in the computational cost columns.

\begin{table}[ht]
\centering
\begin{tabular}{|c|c|c|c|c|c|c|c|}
\hline
%\mulcol{2}{|c|}{&$2^{nd}}
$h$& $n$ &\mulcol{2}{|c|}{$1^{st}$ order} &  \mulcol{2}{|c|}{$2^{nd}$ order}\\
\hline
 & &error in $\tau$  & time(work) & error in $\tau$ & time(work)\\
\hline
1/40 & $161\times321$ &[6.15e-03, 3.86e-03]       &  0.05s(205) &    [1.60e-04, 5.94e-05]&     0.05s(236)\\
\hline
1/80 & $321\times641$  &     [3.07e-03, 1.93e-03]  &    0.21s (221)  &   [3.85e-05, 1.56e-05]    &     0.20s (206)\\
\hline
1/160 & $641\times1281$&      [1.54e-03, 9.67e-04]  &    0.87s (223)    &   [1.08e-05, 4.03e-06]    &     0.86s (221)\\
\hline
1/320& $1281\times2561$  &   [7.68e-04, 4.83e-04]     &  3.80s (259)  &  [3.18e-06, 1.04e-06]  &   3.88s (265)\\
\hline
1/640 &$2561\times5121$&     [3.84e-04, 2.42e-04]   &    16.2s (275)  &  [9.59e-07, 2.66e-07] &    16.6s (280)\\
\hline
1/1280 &$5121\times10241$&     [1.92e-04, 1.21e-04]       &  73.8s (304)     &   [2.99e-07, 6.88e-08]     &  74.6s (307) \\
\hline
\end{tabular}
\caption{Results for the 2D Gaussian factor test case (test case 3). The error measures are in the $[l_\infty, \mbox{mean } l_2]$ norms.}\label{tab:test3-2D}
\end{table}

We now wish to better illustrate the difference between the accuracy of the first order scheme and the second order scheme. First, Figure \ref{fig:Contours} shows contours of the exact and approximate solutions in certain regions of the domain for the second and third test cases. It is clear that the first order approximation is less accurate than the second order one. Next, Figure \ref{fig:LogLog2D} shows plots of the errors in Tables \ref{tab:test1-2D}-\ref{tab:test3-2D} in logarithmic scales for both $h$ and the error norms, where the order of convergence determines the slope of the lines. It is clear that in all cases, using the second order scheme we get second order convergence in mean $l_2$ norm, and a bit more than first order convergence in the $l_\infty$ norm.

\begin{figure}[ht!]
\begin{center}
  \subfigure[Constant gradient of velocity.]{\includegraphics[width=0.32\linewidth]{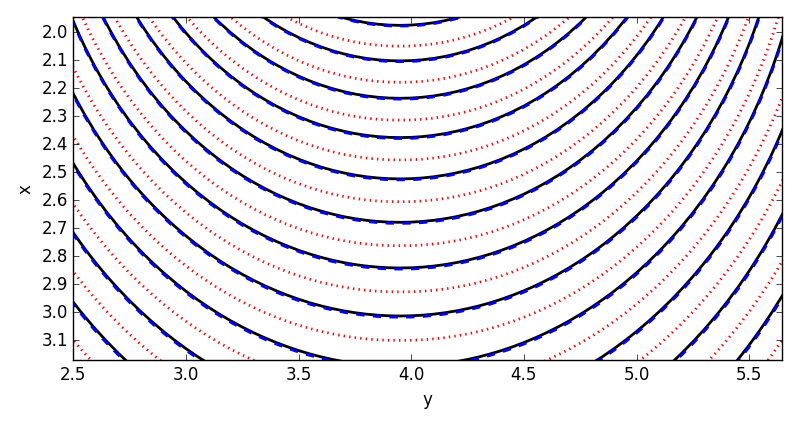}}
  \subfigure[Gaussian factor.]{\includegraphics[width=0.32\linewidth]{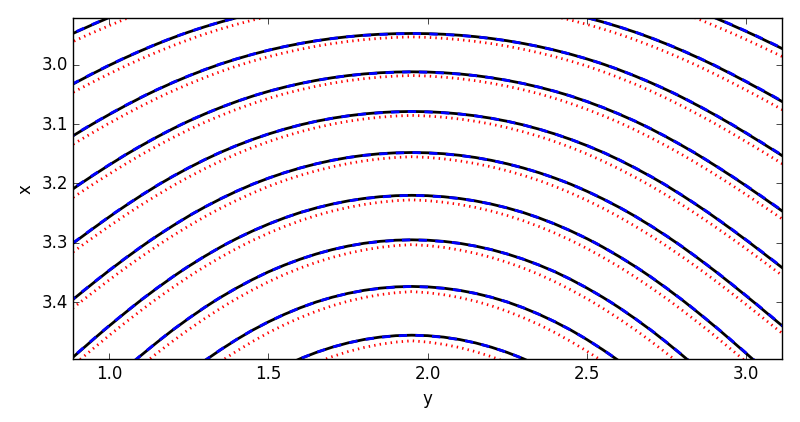}}
\end{center}
\caption{Contours of small regions in the exact, first order accurate and second order accurate travel times $\tau$ using $h=0.1$. The exact solution appears in black line. The first order approximation appears in dotted red line and the second order approximation appears in a dashed blue line mostly right with the exact solution.}
\label{fig:Contours}
\end{figure}

\begin{figure}[ht]
\begin{center}
  \subfigure[2D constant gradient of the squared slowness.]{\includegraphics[width=0.3\linewidth]{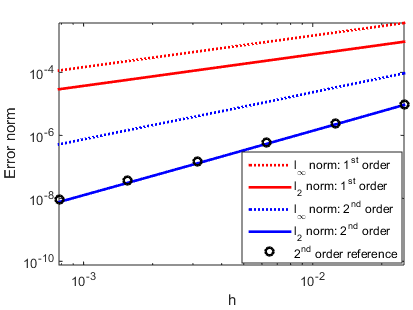}}
  \subfigure[2D constant gradient of velocity.]{\includegraphics[width=0.3\linewidth]{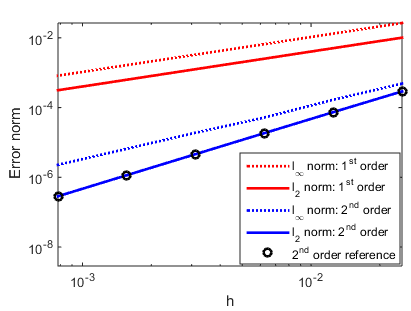}}
  \subfigure[2D Gaussian factor.]{\includegraphics[width=0.3\linewidth]{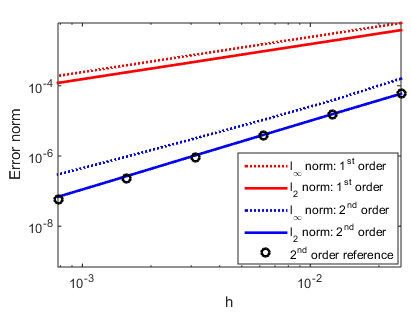}}
\end{center}
\caption{The accuracy of the FM approximations in logarithmic scales for the 2D cases. Red plots are used for first order approximations, blue plots for second order approximations; dotted lines for $l_\infty$ error norm and solid for mean $l_2$ norm. Black circles denote a reference for exact second order convergence rate.}
\label{fig:LogLog2D}
\end{figure}

\subsection{Three dimensional tests}
We now show results for the same type of tests in three dimensions. For all the 3D tests we choose the domain to be $[0,0.8] \times [0,1.6] \times [0,1.6]$, and $h=h_x=h_y=h_z$ varies.
\paragraph{Test case 1: Constant gradient of squared slowness} For the 3D version of this test case we use the parameters $a  = -1.65$, and $s_0 = 2.0$, and the source is located at $(0,0.8,0.8)$. Table \ref{tab:test1-3D} summarizes the results for this test case. Again, like in two dimensions, the first order version of FM yields first order convergence rate in both error norms. When using the second order scheme we get a super-linear convergence rate in the $l_\infty$ column, and second order convergence in the mean $l_2$ column. We note that in 3D the FM algorithm reverts to first order scheme on 2D manifolds, where the derivative in each dimension switches sign, and not on 1D curves as in 2D.

In terms of computational cost, it is obvious that the 3D problem is much more expensive than the 2D one. The computational cost in seconds per grid-point in 3D is about 3 times higher than the corresponding cost in 2D. That is because the treatment of each grid point is more expensive (more neighbors and more derivative directions), and the number of grid points that are processed inside the heap is much larger (a 2D manifold of points compared to a 1D curve). As a result, when we normalize the timing by the cost of a 3D ``work-unit'' (evaluation of \eqref{eq:eikonal} in 3D), the cost grows a little when the mesh-size grows. Still, solving the problem requires 200-500 work units. Again, using the first and second order schemes requires similar computational effort in our 3D implementation of the FM algorithm.

\begin{table}[ht]
\centering
\begin{tabular}{|c|c|c|c|c|c|c|c|}
\hline
%\mulcol{2}{|c|}{&$2^{nd}}
$h$& $n$ &\mulcol{2}{|c|}{$1^{st}$ order} &  \mulcol{2}{|c|}{$2^{nd}$ order}\\
\hline
 & &error in $\tau$  & time(work) & error in $\tau$ & time(work)\\
\hline
1/20 & $17\times33\times33$&   [5.41e-03, 1.46e-03]   &     0.04s(236) & [5.63e-4 ,1.49e-04]     & 0.04s(234)  \\
\hline
1/40 & $33\times65\times65$&    [2.64e-03, 7.05e-04]   &      0.30s (230) &  [2.00e-04 ,3.52e-05]  &0.32s (235) \\
\hline
1/80 & $65\times129\times129$&    [1.30e-03, 3.46e-04]  &     2.88s (332)  & [6.99e-05 ,7.82e-06]  & 2.90s (334) \\
\hline
1/160& $129\times257\times257$&   [6.41e-04, 1.72e-04]   &      28.7s (427)&  [2.51e-05 ,1.68e-06] &   29.0s (432) \\
\hline
1/320 &$257\times513\times513$&    [3.19e-04, 8.55e-05]   &     264s (481)& [8.78e-06 ,3.53e-07]   &    272s (497)  \\
\hline
\end{tabular}
\caption{Results for the 3D constant gradient of squared slowness test case (test case 1). The error measures are in the $[l_\infty, \mbox{mean } l_2]$ norms.}\label{tab:test1-3D}
\end{table}

\paragraph{Test case 2: Constant gradient of velocity}
For the 3D version of this test case we use the parameters $a  = 1.0$, and $s_0 = 2.0$, and the source is located at $(0,0.8,0.8)$. Table \ref{tab:test2-3D} summarizes the results for this test case. As in the previous case, we get first order convergence when using the first order scheme, in both error norms. Again, when using the second order scheme we get a super-linear convergence rate in the $l_\infty$ column, and second order convergence in the mean $l_2$ column.
\begin{table}[ht]
\centering
\begin{tabular}{|c|c|c|c|c|c|c|c|}
\hline
%\mulcol{2}{|c|}{&$2^{nd}}
$h$& $n$ &\mulcol{2}{|c|}{$1^{st}$ order} &  \mulcol{2}{|c|}{$2^{nd}$ order}\\
\hline
 & &error in $\tau$  & time(work) & error in $\tau$ & time(work)\\
\hline
1/20 & $17\times33\times33$ &  [1.35e-02, 5.04e-03]   &   0.04s(237)  &[2.34e-03, 9.36e-04]   &    0.04s(255)   \\
\hline
1/40 & $33\times65\times65$  & [6.24e-03, 2.44e-03]    &  0.31s (234)  & [5.12e-04, 1.72e-04]    &    0.32s (236) \\
\hline
1/80 & $65\times129\times129$&   [3.00e-03, 1.20e-03]   &   2.86s (330)   &[1.70e-04, 3.82e-05]   &      2.89s (334)    \\
\hline
1/160& $129\times257\times257$  &   [1.47e-03, 5.99e-04] &    27.6s (411)&  [5.42e-05, 9.33e-06]  &     28.9s (430) \\
\hline
1/320 &$257\times513\times513$&  [7.30e-04, 2.99e-04]   &  263s (481)& [1.95e-05, 2.29e-06]      & 271s (496)  \\
\hline
\end{tabular}
\caption{Results for the 3D constant gradient of velocity test case (test case 2). The error measures are in the $[l_\infty, \mbox{mean } l_2]$ norms.}\label{tab:test2-3D}
\end{table}

\paragraph{Test case 3: Gaussian factor}
For this 3D test case we use the parameters $\Sigma = \diag(0.2,0.4,0.1)$, $\vec{x}_1 = (0.4,\frac{1.6}{3},0.4)$ (floored to the closest grid point), and the source is located in the point $\vec{x}_0 = (0.2,0.4,0.4)$. Table \ref{tab:test2-3D} summarizes the results for this test case. The results are similar to the previous test case in both the convergence (first/second order using $l_\infty$/$l_2$ norms) and computational costs in seconds and work units.

\begin{table}[ht]
\centering
\begin{tabular}{|c|c|c|c|c|c|c|c|}
\hline
%\mulcol{2}{|c|}{&$2^{nd}}
$h$& $n$ &\mulcol{2}{|c|}{$1^{st}$ order} &  \mulcol{2}{|c|}{$2^{nd}$ order}\\
\hline
 & &error in $\tau$  & time(work) & error in $\tau$ & time(work)\\
\hline
1/20 & $17\times33\times33$ &    [7.53e-03, 3.26e-03] &  0.04s (230) & [3.65e-04, 1.27e-04]     &    0.04s (229)   \\
\hline
1/40 & $33\times65\times65$  &   [3.69e-03, 1.56e-03]  &  0.33s (245)  & [9.95e-05, 2.85e-05]     &  0.34s (253) \\
\hline
1/80 & $65\times129\times129$&    [1.83e-03, 7.62e-04] & 2.77s (319) &  [3.22e-05, 7.50e-06]    &  2.80s (323)  \\
\hline
1/160& $129\times257\times257$  &    [9.11e-04, 3.77e-04]  &  26.4s (393)&  [1.06e-05, 2.06e-06]   &  27.2s (405)\\
\hline
1/320 &$257\times513\times513$&    [4.54e-04, 1.87e-04]  & 267s (487) &[3.54e-06, 5.66e-07]   &  276s (504)  \\
\hline
\end{tabular}
\caption{Results for the 3D Gaussian factor test case (test case 3). The error measures are in the $[l_\infty, \mbox{mean } l_2]$ norms.}\label{tab:test3-3D}
\end{table}

Again we wish to demonstrate the order accuracy of the FM approximations using the first and second order schemes. Figure \ref{fig:LogLog3D} shows the results in Tables \ref{tab:test1-3D}-\ref{tab:test3-3D} in logarithmic scales. Like in 2D, we observe second order convergence rate when the error is measure in mean $l_2$ norm. However, because the second order stencil reduces to first order stencil in two dimensional manifolds, the error in $l_\infty$ norm is higher in 3D than it is in 2D.
\begin{figure}[ht!]
\begin{center}
  \subfigure[3D constant gradient of the squared slowness.]{\includegraphics[width=0.32\linewidth]{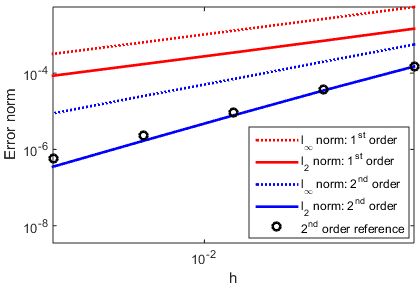}}
  \subfigure[3D constant gradient of velocity.]{\includegraphics[width=0.32\linewidth]{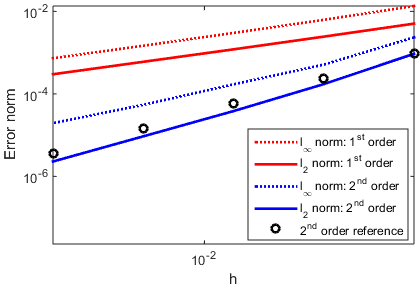}}
  \subfigure[3D Gaussian factor.]{\includegraphics[width=0.32\linewidth]{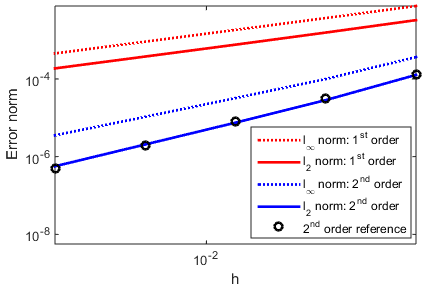}}
\end{center}
\caption{The accuracy of the FM approximations in logarithmic scales for the 3D cases. Red and blue plots are used for first and second order approximations, respectively; dotted lines for $l_\infty$ error norm and solid for mean $l_2$ norm. Black circles denote an exact second order convergence rate.}
\label{fig:LogLog3D}
\end{figure}

\section{Numerical results: travel time tomography.}
In this section we demonstrate a solution of travel time tomography using synthetic travel time data $\bfd_{\obs}$ for a 2D and SEG/EAGE salt model given in \cite{aminzadeh19973} and presented in Figure \ref{fig:Mtrue}, using a $256\times128$ grid that represents an area of approximately $13.5km\times4.2km$. We choose 51 equally distanced sources locations on the open surface (that is, they are located every 5 pixels on the top row), and 256 receivers (located in every pixel on the top row). We note that to have a reasonable solution using the first arrivals for the inverse problem under this setup, the velocity in the interior has to be larger than that on the surface. This is to
guarantee that the first arrival rays obtained on the surface actually come from the interior
but not only travel along the surface. To $\bfd_{\obs}$ we add white Gaussian noise with standard deviation of $0.01\times \mbox{mean}(|\bfd_{\obs}|)$.

\begin{figure}
\begin{center}
  \subfigure[True velocity model.]{\includegraphics[width=0.32\linewidth]{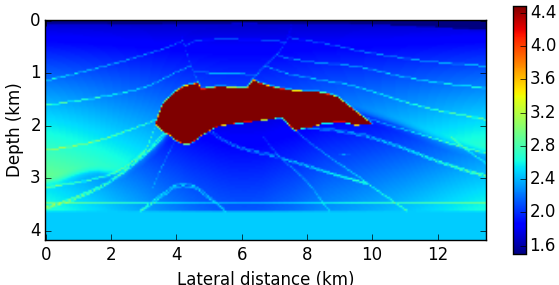}\label{fig:Mtrue}}
  \subfigure[Initial velocity model $\bfm_{ref}$]{\includegraphics[width=0.32\linewidth]{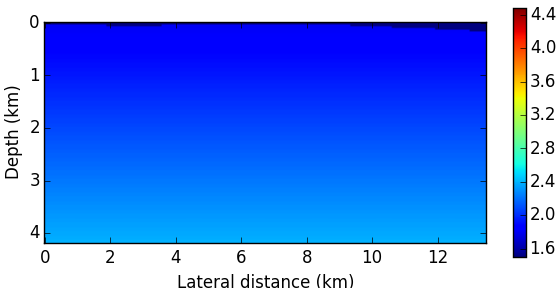}\label{fig:Mref}}
  \subfigure[Recovered velocity model.]{\includegraphics[width=0.32\linewidth]{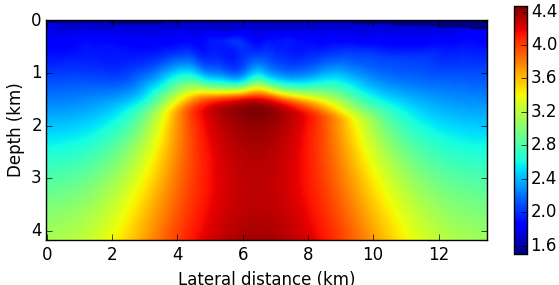}\label{fig:Mest}}
\end{center}
\caption{2D travel time tomography experiment, grid size $256\times128$. Velocities are given in $km/sec$.}
\label{fig:TTT}
\end{figure}

To fit the model to the data, we minimize \eqref{eq:inverseEik} using Gauss-Newton (we perform 10 iterations, where in each we apply 8 CG steps for the Gauss Newton direction problem). We use the general-propose inversion package \cite{jInv16}, which is freely available in \url{https://github.com/JuliaInv/jInv.jl}, together with our FM package mentioned earlier. For that, we first generate an initial slowness model $\bfm^{(0)} = \bfm_{ref}$, whose velocity model shown in Figure \ref{fig:Mref}. This corresponds to a velocity field with a constant gradient in the $y$ direction, similarly to the model in Figure \ref{fig:kappa2}. To bound $\bfm$ from above and from below throughout the minimization, we invert for an auxiliary variable $\bfm'$ and use the following scalar bounded bijective mapping that prevents $\bfm$ from being below $m_L$ or above $m_H$:
$$
m_{bound}(m') = \frac{m_H-m_L}{2}\cdot\tanh\left(\frac{2}{m_H-m_L}\cdot\left(m'-\frac{m_H+m_L}{2}\right)+1\right) + m_L.
$$
That is, instead of minimizing \eqref{eq:inverseEik} as is, we minimize
\begin{eqnarray}
\label{eq:inverseEikBound}
\min_{\bfm'}\phi(\bfm') = \min_{\bfm'}\left\{\sum_{i=1}^{n_s}{\|\bfP^{\top} \bftau^{i}(m_{bound}(\bfm')) - \bfd_{\obs}^{i}\|^2} + \alpha R(\bfm')\right\},
\end{eqnarray}
subject to the same constraints in \eqref{eq:eik_i}. For the regularization $R$ use a simple discrete central-differences Laplacian, and apply it for $\bfm'$; that is
$$R(\bfm') = \frac{1}{2}(\bfm'-\bfm'_{ref})^\top\Delta_h(\bfm'-\bfm'_{ref}),$$
where $\bfm'_{ref}$ is the model such that $m_{bound}(\bfm_{ref}') = \bfm_{ref}$. For $\Delta_h$ we use Neumann boundary conditions, since those lead to an effect of an automatic salt flooding, which is a popular way to treat salt bodies. We set the regularization parameter to be $\alpha=0.5$.
We note that other choices of $\bfm_{ref}$ and regularization terms may definitely be suitable here, but are beyond the scope of this paper. Figure \ref{fig:Mest} shows the result model of the Gauss Newton minimization, and Figure \ref{fig:data} presents the initial and final data and data residuals. In particular, Figure \ref{fig:resPred} shows that the final residual mostly contains the added Gaussian noise. Figure \ref{fig:his} shows that the misfit was indeed reduced throughout the iterations, until the reduction stalls and the misfit reflect the noise level.

\begin{figure}
\begin{center}
  \subfigure[Initial data using $\bfm_{ref}$.]{\includegraphics[width=0.48\linewidth]{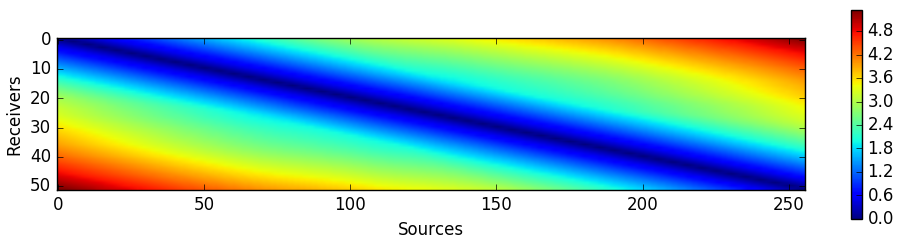}\label{fig:d0}}
  \subfigure[Final predicted data.]{\includegraphics[width=0.48\linewidth]{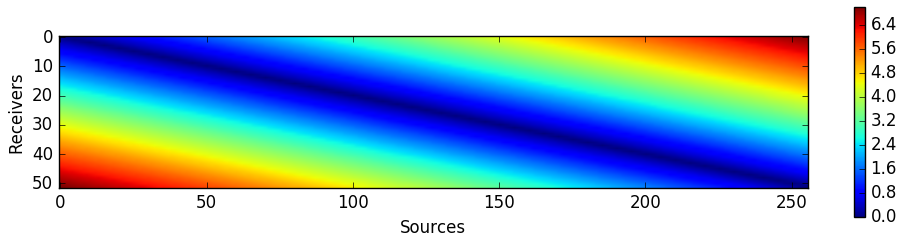}\label{fig:dpred}}
  \subfigure[Initial residual: $|\bfd_{\sf {_{ref}}} - \bfd_{\obs}|$]{\includegraphics[width=0.48\linewidth]{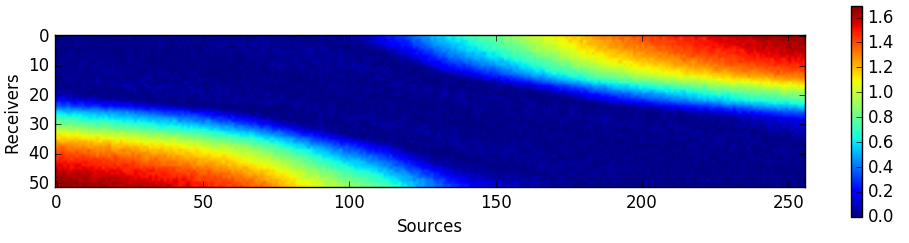}\label{fig:res0}}
  \subfigure[The final residual: $|\bfd_{\sf {_{pred}}} - \bfd_{\obs}|$]{\includegraphics[width=0.48\linewidth]{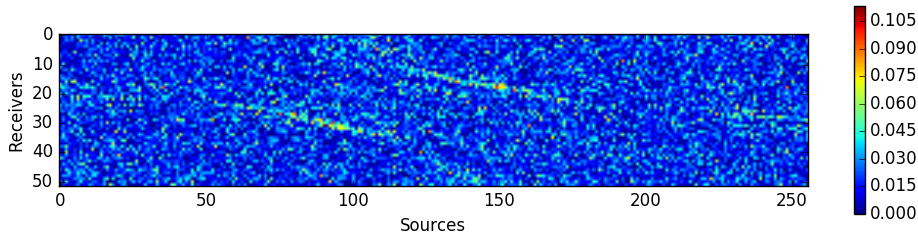}\label{fig:resPred}}
\end{center}
\caption{Initial and final data and residuals in the source-receiver domain.}
\label{fig:data}
\end{figure}

\begin{figure}
\begin{center}
  \includegraphics[width=0.32\linewidth]{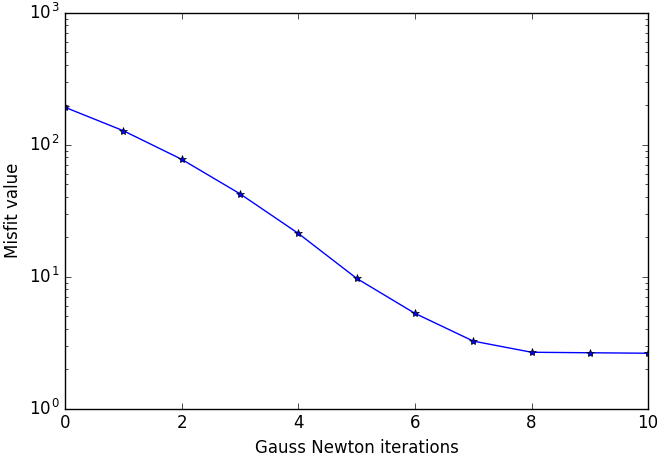}
\end{center}
\caption{Convergence history of the inversion. }
\label{fig:his}
\end{figure}

To demonstrate our algorithm in 3D, we use a 3D version of the same SEG/EAGE model, presented in Figure \ref{fig:Mtrue3D}, using a $256\times256\times128$ grid that represents a volume of $13.5km\times13.5km\times4.2km$. We choose 144 equally distanced sources locations on the open surface, located every 23 pixels on the top surface, and $256\times256$ receivers located on the top surface. We use the same parameters as in the 2D experiment (bound function, regularization, initial 3D model, added noise to the data, number of iterations etc.). Figure \ref{fig:Mest3D} shows the result of the inversion. Similarly to the 2D case, the top part of the model is recovered quite well, while a ``salt flooding'' effect is evident in the bottom part of the model. We performed the inversion using a machine with two Intel(R) Xeon(R) E5-2670 v3 processors with 128 GB of RAM. Using 24 cores, we applied the inversion in approximately 15 hours, and the highest memory footprint of the algorithm reached around 30GB.

\begin{figure}
\begin{center}
  \subfigure[True velocity model.]{\includegraphics[width=0.35\linewidth]{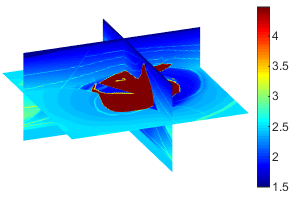}\label{fig:Mtrue3D}}
  \subfigure[Recovered model.]{\includegraphics[width=0.35\linewidth]{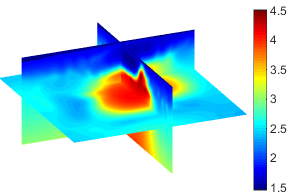}\label{fig:Mest3D}}
\end{center}
\caption{3D travel time tomography experiment, grid size $256\times256\times128$. Velocities are given in $km/sec$.}
\label{fig:TTT3D}
\end{figure}

\section{Conclusions}

In this paper we developed a Fast Marching algorithm for the factored eikonal equation, which in many cases yields a more accurate solution of the travel time than the original equation. Similarly to the original FM algorithm, our version solves the factored problem by exploiting the monotonicity of the solution along the characteristics. Our algorithm is capable of solving the problem using both first and second order schemes. The advantages of our algorithm are (1) its favorable guaranteed $O(n\log n)$ running time, and (2) the easily computed sensitivity matrices for solving the inverse (factored) eikonal equation.

\section{Acknowledgement}
The research leading to these results has received funding from the European Union's - Seventh Framework Programme (FP7/2007-2013) under grant agreement no 623212 - MC Multiscale Inversion.
\section*{References}
{
\bibliographystyle{siam}
\bibliography{../../Helmholtzbib}
}
\end{document}